\documentclass[reqno]{amsart}
\usepackage{amsaddr}
\usepackage{amsmath,amssymb}
\usepackage{nccmath}
\usepackage{graphics,epsfig,color}
\usepackage[mathscr]{euscript}
\usepackage{enumerate}
\usepackage{dsfont}
\usepackage{color}
\usepackage{subfig}
\usepackage{mathtools}
\usepackage{hyperref}
\usepackage{verbatim}
\usepackage{textcomp}

\theoremstyle{plain}
\newtheorem{theorem}{Theorem}[section]

\newtheorem{lemma}[theorem]{Lemma}

\theoremstyle{definition}

\theoremstyle{remark}
\newtheorem{remark}[theorem]{Remark}

\numberwithin{equation}{section}
\numberwithin{theorem}{section}

\newcommand{\upbar}[1]{\,\overline{\! #1}}

\renewcommand{\epsilon}{\varepsilon}

\renewcommand{\tilde}{\widetilde}

\newcommand{\de}{\mathrm{d}}
\newcommand{\ex}{\mathrm{e}}
\newcommand{\ve}{\varepsilon}

\makeatletter
\ifcase \@ptsize \relax
  \newcommand{\miniscule}{\@setfontsize\miniscule{4}{5}}
\or
  \newcommand{\miniscule}{\@setfontsize\miniscule{5}{6}}
\or
  \newcommand{\miniscule}{\@setfontsize\miniscule{5}{6}}
\fi
\makeatother

\title[The interpolation method]{Remarks on the interpolation method}

\author{Roberto Boccagna}
\address{\small{Universit{\`a} dell'Aquila\\ Via Vetoio, Loc. Coppito\\ 67010 L'Aquila, Italia.}}

\author{Davide Gabrielli}
\address{\small{Universit{\`a} dell'Aquila\\ Via Vetoio, Loc. Coppito\\ 67010 L'Aquila, Italia.}}

\begin{document}
	
	\maketitle
	\thispagestyle{empty}
		\begin{abstract}

	We discuss a generalization of the conditions of validity of the interpolation
	method for the density of quenched free energy of mean field spin glasses. The condition
	is written just in terms of the $L^2$ metric structure of the Gaussian random variables. As an example of application we deduce the existence of the thermodynamic limit for a GREM model with infinite branches for which the classic conditions of validity fail.
	
	\bigskip
	
	
	\end{abstract}

\section{Introduction}

The interpolation method is a simple but powerful technique used to prove inequalities for Gaussian random vectors (see for example \cite{JPP} and \cite{K}). This method has great relevance in the field of Mathematical and Theoretical Physics since it represents an essential ingredient in the study of mean field spin glasses. In the breakthrough paper \cite{GT} it has been used to prove the existence of the thermodynamic limit for the quenched density of free energy for the Sherrington-Kirkpatrick model. This was a longstanding problem and its solution was the turning point towards the proof of the Parisi Formula \cite{T}.

Spin glasses are simple mathematical models for disordered systems whose rigorous analysis is indeed a \emph{challenge for mathematicians}. We refer to \cite{P}, \cite{Tb} the mathematically interested reader and to \cite{MPV} the physically interested one. Among plenty of models, one of the most studied is that introduced by Sherrington and Kirkpatrick in \cite{SK} as a solvable elementary model. Indeed the structure of the solution turned out to be much more rich and complex than expected and was build up in a series of papers by Parisi (see \cite{MPV} for a detailed discussion). A rigorous proof of the Parisi conjectured solution was missing for a long time and the interpolation method played a key role in its proof. See \cite{G} for a review on this.

Using the same idea of \cite{GT}, the authors of \cite{CONT} proposed a general setting for the interpolation method in the framework of mean field spin glasses. Furthermore, they successfully applied this technique to prove the existence of the thermodynamic limit for the \textit{Generalized Random Energy Model} (GREM, a family of models introduced in \cite{DERR}) with a finite number of levels.

The ``classical'' hypothesis under which the interpolation method can be applied to the quenched free energy of mean field spin glasses consists of a collection of equalities and inequalities for the covariance matrix of the underlying multivariate Gaussian process. We show that less restrictive conditions are actually needed. More precisely, we show that the method works under conditions that involve just the $L^2$ metric structure of the Gaussian random vectors.
By the correspondence in \cite{S}, \cite{IEEE} this is always an Euclidean metric structure. A condition of this type is very natural since the quenched free energy depends on the distribution of the Gaussian random vector only through its metric structure. This gives an interesting geometric flavor and interpretation.
A similar inequality was obtained through a tricky computation in the framework of Sudakov-Fernique inequalities in \cite{C}. Here we show that the result follows by a general argument involving the special form of the function and that, in a sense, is the best possible. As an example of application of the improved conditions, we consider a GREM model with infinite levels and deduce the existence of the thermodynamic limit for the quenched density of free energy. Indeed, in this case the usual conditions of validity of the interpolation method used in \cite{GT}, \cite{CONT} fail. We can deduce therefore the existence of the thermodynamic limit directly using the simple
argument of the interpolation method. We refer to \cite{R}, \cite{BK1} and \cite{BK2} for the beautiful mathematics involved in the limit of such kind of models.

\smallskip

The structure of the paper is the following.

In Section 2 we briefly recall the basics of the interpolation method together with the conditions used in \cite{GT} and \cite{CONT}; we then discuss the Euclidean metric structure associated to any Gaussian random vector and finally show our generalized conditions.

In Section 3 we discuss two examples. The first one is the Sherrington-Kirkpatrick model. This is done simply to recall the basic mechanism and idea of application.
The second example is a GREM model with infinite levels for which it is necessary to use our generalized conditions to prove the existence of the thermodynamic limit.

In the Appendix we collect some elementary Lemmas.

\section{The interpolation method}
\label{sec:particle systems}

\subsection{The interpolation method}

Let $X=(X_1,\dots ,X_n)$ be a $n$-dimensional zero mean Gaussian random vector having covariance matrix $C$. The $n\times n$ symmetric matrix $C$ is non-negative definite and the elements are defined by $C_{i,j}\coloneqq\mathbb E\left[X_iX_j\right]$. When $C$ is positive definite then the distributions of $X$
is absolutely continuous with respect to the Lebesgue measure on
$\mathbb R^n$ and the density is
\begin{equation}
\phi_{C}\left(x\right)\coloneqq\frac{1}{\sqrt{\left(2\pi\right)^n\textrm{det}\left(C\right)}}\mathrm{e}^{-\frac 12
	(x,C^{-1}x)}\,, \label{den}
\end{equation}
where $\left(\,\cdot\,,\,\cdot\,\right)$ denotes the Euclidean scalar
product in $\mathbb R^n$. We restrict to the case of positive definite matrices
since the other cases can be deduced by a limiting procedure.
We have the Fourier transform representation
\begin{equation}
\phi_C\left(x\right)=\frac{1}{(2\pi)^n}\int_{\mathbb R^n}\de\lambda \,
\ex^{-i(\lambda,x)}\ex^{-\frac 12 (\lambda,C\lambda)}\,. \label{fou}
\end{equation}
We denote by $\mathrm{Tr}\,(\,\cdot\,)$ the trace of a matrix and consider $\upbar{\mathcal{C}}$ the set of non negative definite symmetric matrices  endowed with the Hilbert-Schmidt scalar product
\begin{equation}
\textrm{Tr}\left(AB\right)\,, \qquad A, B\in \upbar{\mathcal C}\,.
\end{equation}
The open set of positive definite symmetric matrices corresponds to $\mathcal C$.

Let $\phi: \mathcal C \times \mathbb
R^n\to \mathbb R^+$ as defined in \eqref{den}. By \eqref{fou} and a direct computation we have
\begin{equation}
\frac{\partial \phi_C\left(x\right)}{\partial C_{i,j}}=\frac{\partial \phi_C\left(x\right)}{\partial C_{j,i}}=\frac{\partial^2\phi_C\left(x\right)}{\partial x_i\partial x_j}\,,
\label{derbas1}
\end{equation}
and
\begin{equation}
\frac{\partial \phi_C\left(x\right)}{\partial C_{i,i}}=\frac
12\frac{\partial^2\phi_C\left(x\right)}{\partial x_i^2}\,.\label{derbas2}
\end{equation}
Recall that in the above formulas $C$ is a symmetric matrix so that the variations in the computation of \eqref{derbas1} are constructed varying symmetrically the matrix $C$. More precisely let $E^{\{i,j\}}$ with $i\neq j$ be the symmetric matrix such that $E^{\{i,j\}}_{i,j}=E^{\{i,j\}}_{j,i}=1$ and having all the remaining elements equal to zero. Given $F:\mathcal C\to \mathbb R$ we define
\begin{equation}
\label{defdersym}
\frac{\partial F\left(C\right)}{\partial C_{j,i}}=\frac{\partial F\left(C\right)}{\partial C_{i,j}}\coloneqq\lim_{\delta\to 0}\frac{F\left(C+\delta E^{\{i,j\}}\right)-F\left(C\right)}{\delta}\,.
\end{equation}

Consider now $f:\mathbb R^n\to \mathbb R$ a $C^2$ function with
\textit{moderate growth at infinity}, for example such that $|f(x)|\leq \ex^{\lambda|x|}$ for a suitable constant $\lambda\geq 0$. This technical condition is related to the validity of some integrations by parts.
We call $\nabla^2f\left(x\right)$ the Hessian matrix of $f$ at $x$, that is the symmetric matrix having elements
\begin{equation}
\left(\nabla^2f\right)_{i,j}\left(x\right)\coloneqq\frac{\partial^2f\left(x\right)}{\partial x_i
	\partial x_j}\,. \nonumber
\end{equation}

The following result is the interpolation method. For the readers convenience we give the short proof.
\begin{lemma}{\bf [Interpolation method]}\label{intmeth}
Consider two mean zero Gaussian random vectors $X, Y$ having covariance matrices respectively given by $C^X$ and $C^Y$. Consider a $C^2$ function $f$ with moderate growth we have
\begin{equation}\label{intf}
\mathbb E\left[f\left(Y\right)\right]-\mathbb E\left[f\left(X\right)\right]=\frac 12 \int_0^1 \de t\,\mathbb E\Big[\mathrm{Tr}\Big(\left(C^Y-C^X\right)\nabla^2f\left(Z\left(t\right)\right)\Big)\Big],
\end{equation}
where
\begin{equation}\label{tutututu}
Z(t)=\sqrt{t}X+\sqrt{(1-t)}Y\,,
\end{equation}
and $X,Y$ are two independent copies of the random vectors.
\end{lemma}
\begin{proof}
When $Z$ is a $n$-dimensional centered Gaussian random vector, then
$\mathbb{E}\left[f(Z)\right]$ depends only on the covariance matrix $C$ of the
vector $Z$. Fix a $C^2$ function $f$ and define the function
$F:\overline{\mathcal C}\to \mathbb R$ as
\begin{equation}
F\left(C\right)\coloneqq\mathbb E\left[f\left(Z\right)\right]\,. \label{dacitare2}
\end{equation}
With the help of formulas \eqref{derbas1}, \eqref{derbas2}, when
$C\in \mathcal C$ we can compute
\begin{eqnarray}
\frac{\partial F\left(C\right)}{\partial C_{i,j}} &=& \int_{\mathbb
	R^n}\de x\,\frac{\partial \phi_{C}\left(x\right)}{\partial C_{i,j}}f\left(x\right)
=\int_{\mathbb R^n}\de x\,\frac{\partial^2
	\phi_{C}\left(x\right)}{\partial x_i\partial
	x_j}f\left(x\right)\\
&=&\int_{\mathbb R^n}\de x\,
\phi_{C}\left(x\right)\frac{\partial^2f\left(x\right)}{\partial x_i\partial
	x_j}=\mathbb E\left[\left(\nabla^2
f\right)_{i,j}\left(Z\right)\right],\label{der1}
\end{eqnarray}
and
\begin{eqnarray}
\frac{\partial F\left(C\right)}{\partial C_{i,i}}&=& \int_{\mathbb
	R^n}\de x\,\frac{\partial \phi_{C}\left(x\right)}{\partial C_{i,i}}f\left(x\right)
=\frac 12\int_{\mathbb R^n}\de x\,\frac{\partial^2
	\phi_{C}\left(x\right)}{\partial x_i^2}f\left(x\right)\nonumber \\
&=&\frac 12\int_{\mathbb R^n}\de x\,
\phi_{C}\left(x\right)\frac{\partial^2f\left(x\right)}{\partial x_i^2}=\frac 12\, \mathbb
E\left[\left(\nabla^2f\right)_{i,i}\left(Z\right)\right]\,.\label{der2}
\end{eqnarray}

Given a $C^1$ parametric curve $\left\{C(t)\right\}_{t\in[0,1]}$ on $\mathcal C$
such that $C(0)=C^X$ and $C(1)=C^Y$, then we have
\begin{equation}
\mathbb E\left[f\left(Y\right)\right]-\mathbb E\left[f\left(X\right)\right]=\frac 12\int_0^1 \de t\,\mathbb
E\left[\mathrm{Tr}\left(\frac{\de C\left(t\right)}{\de t}\nabla^2f\left(Z\left(t\right)\right)\right)\right]\,,
\end{equation}
where $Z\left(t\right)$ is a centered Gaussian random vector having covariance
$C\left(t\right)$.  The special case when the curve linearly interpolates
between $C^X$ and $C^Y$ gives \eqref{intf} with $Z(t)$  given by \eqref{tutututu}.
If one or both the matrices $C^X$ and $C^Y$ are not strictly positive definite,
it is possible to add to the matrices $\epsilon \mathbb I$, do the same computation as above and finally take the limit $\epsilon \to 0$.
\end{proof}

The above formula is the core of the interpolation method. It is
very useful to establish inequalities between the two expected
values on the left hand side of \eqref{intf}.



\smallskip

The \emph{Guerra-Toninelli interpolation method} is a simple but powerful
technique developed in the study of mean field spin glasses (see \cite{G, GT} and
references therein), which is based on an abstract theorem
about Gaussian random variables. It corresponds to the interpolation
method \ref{intmeth} with the special choice
of the function
\begin{equation}\label{interpol}
f(x)=\log \sum_{i=1}^nw_i \ex^{x_i}\,,
\end{equation}
where $w_i\in \mathbb R^+$ are some fixed positive weight.

In particular, Guerra and Toninelli obtained and used the following result
(this is Theorem 2 in \cite{G}) to prove the existence of the thermodynamic limit of the Sherrington-Kirkpatrick model. The same idea and the same Theorem, (Theorem \ref{loro} below) was used later on in \cite{CONT} to deduce the existence of the thermodynamic limit for
a GREM model \cite{DERR} with a finite number of levels.

\begin{theorem}\label{loro}
Let $X,Y$ two centered Gaussian random vectors and the function $f$ given by \eqref{interpol}.
If
\begin{align}
& C^X_{i,i}=C^Y_{i,i}\,, \qquad \forall \, i\,,\label{condd1}\\
& C^X_{i,j}\geq  C^Y_{i,j}\,, \qquad \forall \, i\neq j\,,\label{condd2}
\end{align}
then we have
\begin{equation}\label{ladis}
\mathbb
E\left[f\left(Y\right)\right]\geq\mathbb E\left[f\left(X\right)\right]\,.
\end{equation}
\end{theorem}
We show the proof of Theorem \ref{loro} that is based on the interpolation formula
\eqref{intf}.
\begin{proof}[Proof of Theorem \ref{loro}]
Let us call, for any $i=1,\ldots,n$
\begin{equation}\label{muu}
\mu_i\left(x\right)\coloneqq\frac{w_i\ex^{x_i}}{\sum_{j=1}^nw_j\ex^{x_j}}\,.
\end{equation}
By a direct computation, when $f$ is \eqref{interpol}, we have
	\begin{align}
	\frac{\partial^2f\left(x\right)}{\partial x_i^2}=\mu_i\left(x\right)-\mu_i^2\left(x\right)\,,\label{dergt1}\\
	\frac{\partial^2f\left(x\right)}{\partial x_i\partial
		x_j}=-\mu_i\left(x\right)\mu_j\left(x\right)\,.\label{dergt2}
	\end{align}
By the formulas \eqref{dergt1}, \eqref{dergt2} and conditions \eqref{condd1}, \eqref{condd2}, we have that
	\begin{equation}
	(C^Y-C^X)_{i,j}\left(\nabla^2f\right)_{i,j}\left(x\right)\geq 0\,, \ \ \
	\ \ \forall\, x\in \mathbb R^d, \;\;\; \forall \,i,j\label{cond}\,
	\end{equation}
	and the result follows by \eqref{intf}.
\end{proof}

\subsection{Covariances and metrics}

We start showing a simple but useful Lemma
\begin{lemma}
We have that the $n\times n$ symmetric matrix $C$ belongs to $\mathcal C$ if and only if there exist $n$ vectors $a^{(i)}\in \mathbb R^n$ such that
\begin{equation}\label{defcv}
C_{i,j}=\big(a^{(i)},a^{(j)}\big)\,.
\end{equation}
\end{lemma}
\begin{proof}
If a matrix can be written like \eqref{defcv} then we have
$$
\sum_{i,j}x_iC_{i,j}x_j=|v|^2\geq 0\,,\qquad \forall \,x_1, \dots ,x_n\,,
$$
where $v=\sum_{i=1}^nx_ia^{(i)}$.
Conversely, if $C \in \mathcal C$ then it can be written as $C=AA^T$ where $A$ is a $n\times n$ matrix and $A^T$ denotes its transposed matrix.
Let us introduce the vectors $a^{(i)}=(a^{(i)}_1,\dots,a^{(i)}_n)\in \mathbb R^n$ with $i=1,\dots ,n$ defined by $a^{(i)}_j\coloneqq A_{i,j}$. In terms of these vectors we have \eqref{defcv}.
\end{proof}

A finite metric space with $n$ points is called \emph{Euclidean} if there exists a collection of $n$ points on $\mathbb R^k$ having the same distances. Of course we can always fix $k=n$. Not every metric space can be realized in this way. The simplest example is the minimal path metrics on the vertices of the graph in Figure \ref{esempio} where the edges have all length 1.

\begin{figure}[htbp!]
	\centering
	\includegraphics{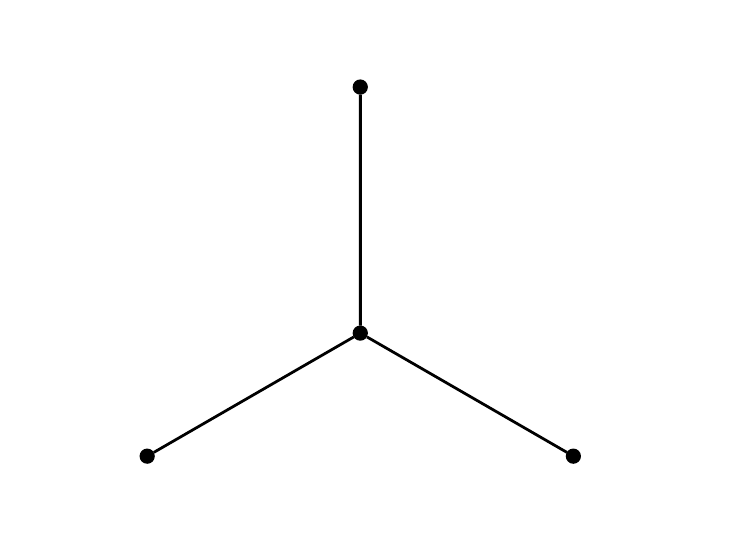}
	\caption{The simplest non Euclidean metric space. }\label{esempio}
\end{figure}

Given a centered Gaussian random vector $X$ there is naturally associated the metric $d_X$ that is the $L^2$ distance between the random variables
\begin{equation}\label{ladist}
d_X(i,j)\coloneqq \sqrt{\mathbb E\left[\left(X_i-X_j\right)^2\right]}=\sqrt{C^X_{i,i}+C^X_{jj}-2C^X_{i,j}}.
\end{equation}

We have the following result (see for example \cite{IEEE,S})
\begin{lemma}
A finite metric space $\left(\left\{1,\dots,n\right\}, d\right)$ is Euclidean if and only if there exists a mean zero Gaussian random vector $X=(X_1,\dots ,X_n)$
such that $d=d_X$.
\end{lemma}
\begin{proof}
Consider $d$ and Euclidean distance and let $a^{(i)}$, $i=1,\dots ,n$ be some points on $\mathbb R^n$ that realize such a distance. Let $A$ be an $n\times n$ matrix
defined  by $A_{i,j}\coloneqq a^{(i)}_j$. Let $Z=(Z_1,\dots,Z_n)$ be a vector of i.i.d. standard Gaussian random variables and consider the Gaussian vector $X=AZ$ whose
covariance $C^X=AA^T$ coincides with the right hand side of \eqref{defcv}. Using \eqref{ladist} we have
\begin{equation}\label{si}
d_X\left(i,j\right)=\big|a^{(i)}-a^{(j)}\big|=d\left(i,j\right)\,.
\end{equation}
Conversely given $X$ a Gaussian mean zero vector with covariance $C^X$ and let $A$ an $n\times n$ matrix such that $C^X=AA^T$. Define $n$ vectors in $\mathbb R^n$
by $a^{(i)}_j\coloneqq A_{i,j}$; by \eqref{ladist} we have that $d_X$ is determined by the first equality in \eqref{si} and is therefore Euclidean.
\end{proof}

The metric structure $d_X$ contains less information than the covariance $C^X$ and there are random vectors having different covariances
but the same metric structure. This type of invariance is best understood in terms of the vectors in $\mathbb R^d$ using the Lemmas in Appendix
that characterizes invariance by rotations and translations. In particular we can completely
characterize the centered Gaussian random variables that share the same
metric structure.
\begin{lemma}\label{belloteo}
	Given $X$ and $Y$ two $n$-dimensional centered Gaussian random
	vectors, we have that $d^X=d^Y$ if and only if there exists a centered Gaussian random variable $W$ such that
	the random vector $X_i+ W$, $i=1,\dots ,n$ has the same distribution
	of $Y$.
\end{lemma}
\begin{proof}
	If $Y$ has the same distribution of $X+W$ then
	$$
	d_Y(i,j)=\sqrt{\mathbb{E}\left[\left(Y_i-Y_j\right)^2\right]}=\sqrt{\mathbb
	{E}\left[\left(X_i+W-X_j-W\right)^2\right]}=d_X(i,j)\,.
	$$
	Conversely, suppose that $d_X=d_Y$. We have that there exist two
	matrices $A^X$ and $A^Y$ such that $A^XZ$ has the same distribution
	of $X$ and $A^YZ$ has the same distribution of $Y$, where $Z$ is an $n$ vector of i.i.d. standard
	Gaussian random variables. We define two collections $v^{(i)}, w^{(i)}$ $i=1,\dots n$ of vectors in $\mathbb R^n$
	defined by $v^{(i)}_j\coloneqq A^X_{i,j}$ and $w^{(i)}_j\coloneqq A^Y_{i,j}$. Since $d_X=d_Y$ we have
	\begin{equation}
	\big|v^{(i)}-v^{(j)}\big|=\big|w^{(i)}-w^{(j)}\big|\,, \qquad  \forall\, i,j\,,
	\label{pimpa2}
	\end{equation}
	and by Lemma \ref{pimpalemma} there
	exist $O\in O(n)$ and a vector $b\in \mathbb R^n$ such that $w^{(i)}=Ov^{(i)}+b$, $i=1,\dots ,n$.
	In terms of the corresponding matrices this means that $A^Y=A^XO^T+B$,
	where the matrix $B$ is defined as $B_{i,j}\coloneqq b_j$.
	We obtain therefore
	\begin{equation}\label{bonci}
	Y=A^XO^TZ+BZ\,.
	\end{equation}
	The random vector $A^XO^TZ$ is a centered Gaussian random vector with covariance $A^XO^TO(A^X)^T=C^X$ so that it has
	the same law of $X$. The random vector $BZ$ has all the components equal and setting $W=\sum_{j=1}^nb_jZ_j$
	we finish the proof.
\end{proof}
A direct consequence of the above result is the following. Define the function $F:\upbar{\mathcal C}\to \mathbb R$ by
\begin{equation}\label{moi}
F(C)\coloneqq \mathbb E\left[\log \sum_{i=1}^nw_i\ex^{X_i}\right]\,,
\end{equation}
where $X$ is a centered Gaussian random vector with covariance $C$.
\begin{lemma}
Given $C^X, C^Y \in \upbar{\mathcal C}$ such that $d_X=d_Y$, then $F(C^X)=F(C^Y)$.
\end{lemma}
\begin{proof}
Since $d_X=d_Y$ by Lemma \ref{belloteo} we have that $Y=X+W$ and therefore
\begin{eqnarray*}
&F(C^Y) =\mathbb E\left[\log \sum_{i=1}^nw_i\ex^{Y_i}\right]
 =\mathbb E\left[\log \sum_{i=1}^nw_i\ex^{X_i+W}\right]\\
& =\mathbb E\left[W+\log \sum_{i=1}^nw_i\ex^{X_i}\right]=F(C^X)\,,
\end{eqnarray*}	
where the last equality follows by the fact that $W$ is centered.
\end{proof}

This Lemma simply says that we can define the right hand side of \eqref{moi}
as $\tilde F(d)$ since the function depends just on the metric structure of
the random variables and not on their correlations.

We expect therefore to have a version of Theorem \ref{loro}
with conditions written just in terms of the metrics. This is done in the next section.

\subsection{A generalized condition}

We show how to generalize Theorem \ref{loro} proving that \eqref{ladis} can be deduced under weaker hypotheses concerning just the metric structures. The same inequality has been obtained in \cite{C} with a tricky computation. Here we show that this fact follows from a general argument and that it is somehow the best possible bound.

\begin{theorem}\label{nostro}
Let $X,Y$ two centered Gaussian random vectors and the function $f$ given by \eqref{interpol}.
	If
	\begin{equation}
	d_Y\left(i,j\right)\geq d_X\left(i,j\right) \qquad \forall\,
	i,j\,, \label{condGGT}
	\end{equation}
	then
\begin{equation}
\mathbb E\left[f\left(Y\right)\right]\geq\mathbb E\left[f\left(X\right)\right].
\end{equation}
\end{theorem}
Note that if conditions \eqref{condd1} and \eqref{condd2} are satisfied then \eqref{condGGT} holds, but it is easy to construct examples for which \eqref{condGGT} holds
but \eqref{condd1}, \eqref{condd2} are violated.

Observe that for any $x$ we have that $\mu\left(x\right)=\left(\mu_1\left(x\right), \dots ,\mu_n\left(x\right)\right)\in
\mathcal I^n$ (recall definition \eqref{muu}) where
$$
\mathcal I^n=\bigg\{\mu=\left(\mu_1,\dots,\mu_n\right) :\,0\leq \mu_i \leq 1,\,\,
\sum_{i=1}^n\mu_i=1\bigg\}\,.
$$
Namely, $\mathcal I^n\subset \mathbb R^n$ is a $(n-1)$-dimensional
simplex with extremal elements  $\mu^{(1)},\dots ,\mu^{(n)}$, where $
\mu^{(l)}_i=\delta_{li}$.

We start with a preliminary Lemma

\begin{lemma}
	\label{lemme} Consider a symmetric matrix $D$ and the function
	$G:\mathcal I^n\to \mathbb R$ defined as
	\begin{equation}\label{gi}
	G\left(\mu\right)\coloneqq\sum_{i=1}^n\mu_iD_{ii}-\sum_{i=1}^n\sum_{j=1}^n\mu_i\mu_jD_{ij}.
	\end{equation}
	We have that
	\begin{equation}
	\inf_{\mu\in \mathcal I^n}G\left(\mu\right)\geq 0 \label{max}
	\end{equation}
	if and only if
	\begin{equation}
	D_{ii}+D_{jj}-2D_{ij}\geq 0 \qquad \forall\,
	i,j\in\left\{1,\dots,n\right\}. \label{condfin}
	\end{equation}
\end{lemma}
\begin{proof}
	If condition \eqref{condfin} holds, then
	$$
	G(\mu)\geq\sum_{i=1}^n\mu_iD_{ii}-\frac 12
	\sum_{i=1}^n\sum_{j=1}^n\mu_i\mu_j\left(D_{ii}+D_{jj}\right)=0.
	$$
	To obtain the last identity we used the fact that $\mu \in \mathcal
	I^n$.
	Conversely, suppose inequality \eqref{max} to hold. Choose $\mu$ such that $\mu_l=\mu_m=\frac{1}{2}$ for some $l\neq m$ and $0$ otherwise; then \eqref{gi} becomes
	\begin{equation}
	\frac{1}{4}\left(D_{ll}+D_{mm}\right)-\frac{1}{2}D_{lm}\geq 0
	\end{equation}
	where we used the symmetry of $D$. Consider all the couples $l,m\in\left\{1,\ldots,n\right\}$ to get the result.
\end{proof}

\begin{proof}[Proof of Theorem \ref{nostro}]
By formula \eqref{intf} we deduce the results once we show that
\begin{equation}\label{miacond}
\inf_{x\in \mathbb R^n}\bigg\{\mathrm{Tr}\Big(D\left(\nabla^2f\right)_{i,j}(x)\Big)
\bigg\}\geq 0\,,
\end{equation}
where we called
\begin{equation}\label{cucina}
D\coloneqq C^Y-C^X\,.
\end{equation}
Using \eqref{dergt1} and \eqref{dergt2} we obtain that
the expression to be minimized in \eqref{miacond} is
\begin{equation}
\sum_{i=1}^nD_{i,i}\mu_i\left(x\right)-\sum_{i=1}^n\sum_{j=1}^nD_{i,j}\mu_i\left(x\right)\mu_j\left(x\right).
\label{asou}
\end{equation}
We have therefore that the infimum in \eqref{miacond} coincides with
$\inf_{\mu\in \mathcal I^n}G(\mu)$ and the result follows by Lemma \ref{lemme} since \eqref{condfin} with the matrix $D$ defined by \eqref{cucina} coincides
with \eqref{condGGT}.
\end{proof}

\section{Examples}
In this section we discuss two examples, obtaining the existence of the thermodynamic limit for the quenched free energy of two models. The first one is the  Sherrington-Kirkpatrick model. The existence of the thermodynamic limit for this model was obtained, by the interpolation method, in the breakthrough paper \cite{GT}. This was done using the result \ref{loro}. We review this result as a warm-up to fix ideas and the basic constructions. We use however Theorem \ref{nostro} and discuss the result just in terms of the metrics. Then we discuss a class of Generalized Random Energy Models \cite{DERR} for which in general conditions \eqref{condd1}, \eqref{condd2} fail while condition \eqref{condGGT} hold. We refer to \cite{R}, \cite{BK1} and \cite{BK2} for the beautiful mathematics involved in the limit of such kind of models.

\subsection{The  Sherrington-Kirkpatrick model}
The Sherrington-Kirkpatrick model is a mean field spin glass model \cite{G}, \cite{P}, \cite{SK}, \cite{Tb}. Spins configurations are $\sigma\in \{-1,1\}^N$ and the energy of the system is given by
\begin{equation}\label{ESK}
H_N\left(\sigma\right)\coloneqq-\frac{1}{\sqrt N}\sum_{i,j=1}^NJ_{i,j}\sigma\left(i\right)\sigma\left(j\right)\,,
\end{equation}
where $J_{i,j}$ are i.i.d. standard Gaussian random variables. Small variants of the model consider different sums in \eqref{ESK} but all the variants are equivalent modulo simple transformations.  The spins are associated to the vertices of a complete graph and the interaction between each pair of spins is determined by the variables $J$'s.
The partition function is defined as
\begin{equation}
Z_N\left(\beta\right)\coloneqq \sum_{\{\sigma\}}\ex^{-\beta H_N(\sigma)}\,,
\end{equation}
where the parameter $\beta$ is the inverse temperature and the quenched free energy per site is defined by
\begin{equation}
F_N\left(\beta\right)\coloneqq -\frac{1}{\beta N}\mathbb{E}\left[\log Z_N\left(\beta\right)\right]\coloneqq \frac{1}{\beta N}\alpha_N\left(\beta\right)\,,
\end{equation}
where the last equality defines the symbol $\alpha_N\left(\beta\right)$.
The variables $\left(-\beta H_N\left(\sigma\right)\right)_{\sigma\in \left\{-1,1\right\}^N}$ are a centered Gaussian random vector with covariance
\begin{equation}
\beta^2\mathbb E\left[H_N\left(\sigma\right)H_N\left(\sigma'\right)\right]=\frac{\beta^2}{N} \sum_{i,j}\sigma\left(i\right)\sigma\left(j\right)\sigma'\left(i\right)\sigma'\left(j\right)=N\beta^2q_N^2\left(\sigma,\sigma'\right)\,,
\end{equation}
where
\begin{equation}
q_N\left(\sigma,\sigma'\right)\coloneqq \frac{1}{N}\sum_{i=1}^N\sigma\left(i\right)\sigma'\left(i\right)\,,
\end{equation}
is the \emph{overlap} between the configurations $\sigma$ and $\sigma'$.
The corresponding distance according to \eqref{ladist} is given by
\begin{equation}\label{ladistq}
d_N\left(\sigma,\sigma'\right)=
\beta\sqrt{8N\Big[d^H_N\left(1-d^H_N\right)\Big]}\,,
\end{equation}
where
\begin{equation}\label{distsk}
d^H_N\left(\sigma,\sigma'\right)\coloneqq \frac 1N\sum_{i=1}^N\mathbb I\Big(\sigma\left(i\right)\neq \sigma'\left(i\right)\Big)
\end{equation}
is the \textit{Hamming distance}. Notice that we have of course $d_N\left(\sigma,\sigma\right)=0$
but we have also $d_N\left(\sigma,-\sigma\right)=0$ since $H_N\left(\sigma\right)=H_N\left(-\sigma\right)$. The fact that the right hand side of \eqref{distsk} is a distance (indeed a pseudo distance) is not trivial but follows directly since it is obtained by \eqref{ladist}.

\smallskip

Let us split the system into two subsystems $S_1, S_2$ with respectively $N_1$ and $N_2$ vertices with $N_1+N_2=N$. We erase the interaction between spins that belong to different subsystems. We define the restricted Hamiltonians of the subsystems as
\begin{equation}\label{ESKsub}
H_{N_k}(\sigma)\coloneqq -\frac{1}{\sqrt N_k}\sum_{i,j\in S_k}J_{i,j}\sigma\left(i\right)\sigma\left(j\right)\,, \qquad k=1,2\,,
\end{equation}
where we remark that the sum is restricted to the indices belonging to the subsystems labeled $k=1,2$. Here and hereafter we continue to use the symbol $\sigma$ both for the full configuration as well as for the configuration restricted to a subsystem.
When a configuration appears in an expression that is labeled by a subsystem then we mean the configuration restricted to the subsystem. For example $d^H_{N_k}\left(\sigma,\sigma'\right)$ and $d_{N_k}\left(\sigma,\sigma'\right)$ are respectively the Hamming distance \eqref{distsk} and the distance \eqref{ladistq} when the configuration is restricted to the subsystem $k=1,2$. Note that with this notation we have the key relationship
\begin{equation}\label{tir}
d^H_N\left(\sigma,\sigma'\right)=\frac{N_1}{N}d^H_{N_1}\left(\sigma,\sigma'\right)+
\frac{N_2}{N}d^H_{N_2}\left(\sigma,\sigma'\right)\,.
\end{equation}
Another important relationship is
\begin{equation}\label{pesca}
\sum_{\{\sigma\}} \ex^{-\beta\left(H_{N_1}(\sigma)+H_{N_2}(\sigma)\right)}=Z_{N_1}(\beta)Z_{N_2}(\beta)\,.
\end{equation}
We apply Theorem \ref{nostro} with the vectors
$$
\left\{
\begin{array}{l}
Y=\big(-\beta H_N(\sigma)\big)_{\sigma\in \{-1,1\}^N}\,,\\
X=\left( -\beta H_{N_1}(\sigma)-\beta H_{N_2}(\sigma)\right)_{\sigma\in \{-1,1\}^N} \,.
\end{array}
\right.
$$
The condition \eqref{condGGT} becomes the super-Pythagorean relation
\begin{equation}
d_N\geq \sqrt{d_{N_1}^2+d_{N_2}^2}\,,
\end{equation}
that is equivalent to
\begin{equation}
\Big[d^H_N\left(1-d^H_N\right)\Big]\geq \frac{N_1}{N}\Big[d^H_{N_1}\left(1-d^H_{N_1}\right)\Big]+\frac{N_2}{N}\Big[d^H_{N_2}\left(1-d^H_{N_2}\right)\Big]\,.
\end{equation}
The above inequality is true
by \eqref{tir} and the concavity of the real function $x \to x\left(1-x\right)$.
By Theorem \ref{nostro} and \eqref{pesca} we deduce
\begin{equation}\label{keysk}
\alpha_N\left(\beta\right)\leq \alpha_{N_1}\left(\beta\right)+\alpha_{N_2}\left(\beta\right)\,,
\end{equation}
and by sub-additivity and the classic Fekete Lemma we deduce
that the limit of the  quenched free energy per site exists
\begin{equation}
\lim_{N\to\infty}F_N\left(\beta\right)=\lim_{N\to \infty}\frac{1}{\beta N}\alpha_N\left(\beta\right)=\inf_N\frac{1}{\beta N}\alpha_N\left(\beta\right)\,.
\end{equation}

\subsection{The Generalized Random Energy Model}

The \textit{Generalized Random Energy Model} (GREM) is a spin glass model introduced by Derrida \cite{DERR} to generalize the REM (\textit{Random Energy Model}) imposing pair correlations between energies. The model has a hierarchical structure, as any spin configuration correspond to a leaf of a given rooted tree.

We consider sequences of finite trees codified by finite strings of nonnegative integers.
Let $n\in\mathbb{N}$ and $\underline k=\left(k_1,\ldots,k_n\right)$ a vector of nonnegative integers and call $|\underline k| \coloneqq k_1+\ldots+k_n$. The tree $\mathcal{T}_{\underline k}$ is constructed as follows. The root (that is the unique node at level $0$) is connected to $2^{k_1}$ nodes to compose the first level. Each node of the first level is connected to $2^{k_2}$ nodes of the second level; we have therefore $2^{k_1+k_2}$ nodes on the second level and so on. The $n$-th level consists of $2^{k_1}2^{k_2}\ldots 2^{k_n}=2^{|\underline k|}$ leaves. If there exists a $1\leq j<n$ such that $k_j=0$, we mean that the nodes of the level $j$ coincide with those of the level $j-1$. A spin configuration $\sigma\in\left\{-1,1\right\}^{|\underline k|}$ is then attached to each leaf. The Hamiltonian is
\begin{equation}\label{energia-grem}
H_{\underline k}\left(\sigma\right)=-\sqrt{|\underline k|}\left(\ve_1^{(\sigma)}+\ldots+\ve_{n}^{(\sigma)}\right)\,,
\end{equation}
where $\ve_i^{(\sigma)}\sim\mathcal{N}\left(0,a_i\right)$ if $k_i>0$ and $\ve_i^{(\sigma)}=0$ if $k_i=0$. For any $i\in \mathbb N$ we have that the $a_i$'s are  positive numbers such that $\sum_{i=1}^{+\infty} a_i=1$.

The random variables $\ve$'s are attached to the edges of the tree. More precisely
attached to the edges that connect the level $i-1$ to the level $i$ there is a family of i.i.d. centered Gaussian random variables with variance $a_i$, one for each edge. When we write $\ve^{(\sigma)}_i$ we mean then the random variable associated to the unique edge that connects level $i-1$ to level $i$ and that belongs to the unique path from the leaf associated to $\sigma$ to the root. When $k_i=0$ there are no edges from level $i-1$ to level $i$ and therefore we set $\ve^{(\sigma)}_i=0$.
Then, $\left(H_{\underline k}(\sigma)\right)_{\sigma \in \left\{-1,1\right\}^{|\underline k| }}$ is a centered Gaussian random vector on the $|\underline k|$-dimensional hypercube $\left\{-1,1\right\}^{|\underline k|}$.

We call $l=l\left(\sigma,\tau\right)\in \{0,1, \dots n-1\}$ the level of the hierarchy at which the two paths from the leaves $\sigma$ and $\tau$ of $\mathcal{T}_{\underline k}$ to the root merge.
 The two configurations share the same energy variables $\ve^{(\sigma)}_i=\ve^{(\tau)}_i$ for any $i\le l$, while $\ve^{(\sigma)}_i\neq\ve^{(\tau)}_i$ whenever $i> l$. When $\ve^{(\sigma)}_i$ and $\ve^{(\tau)}_i$ are different, they are independent. Furthermore, $\ve^{(\sigma)}_i$ and $\ve^{(\tau)}_j$ are always independent if $j\neq i$. We define $\tilde a_i\coloneqq a_i$ when $k_i>0$ and $\tilde a_i\coloneqq0$ when $k_i=0$. We get
\begin{equation}\label{corsa}
\mathbb{E}\left[H_{\underline{k}}\left(\sigma\right)H_{\underline{k}}\left(\tau\right)\right]=\left|\underline{k}\right|\sum_{i=1}^{l}\tilde a_i\,,
\end{equation}
pointing out that the right hand side above is zero when $l=0$.
The corresponding metric according to \eqref{ladist}
is given by
\begin{equation}\label{malato}
d_{\underline{k}}\left(\sigma,\tau\right)=\sqrt{\mathbb{E}\left[\left(H_{\underline{k}}\left(\sigma\right)-H_{\underline{k}}\left(\tau\right)\right)^2\right]}=\sqrt{ 2 \left|\underline{k}\right|\sum_{i=l+1}^n\tilde a_i}\,.
\end{equation}
The term  inside the square root on the right hand side represents, up to a multiplicative factor, the minimal path length distance between the two leaves $\sigma$ and $\tau$
on the tree when each edge between level $i-1$ and $i$ has a length given by $a_i$.
Since the graph is a tree the path is unique and the metric \eqref{malato} is an ultrametric. We introduce, for notational convenience, the normalized distance
\begin{equation}\label{malato2}
{{s}_{\underline{k}}}\left(\sigma,\tau\right)\coloneqq \sqrt{2\sum_{i=l+1}^n\tilde a_i}\,,
\end{equation}
so that $d_{\underline{k}}\left(\sigma,\tau\right)=\sqrt{\left|{\underline{k}}\right|}s_{\underline{k}}\left(\sigma,\tau\right)$ for any pair of configurations $\sigma$ and $\tau$.

\begin{figure}[htbp!]
	\centering
	\hspace*{-0.45cm}
	\includegraphics{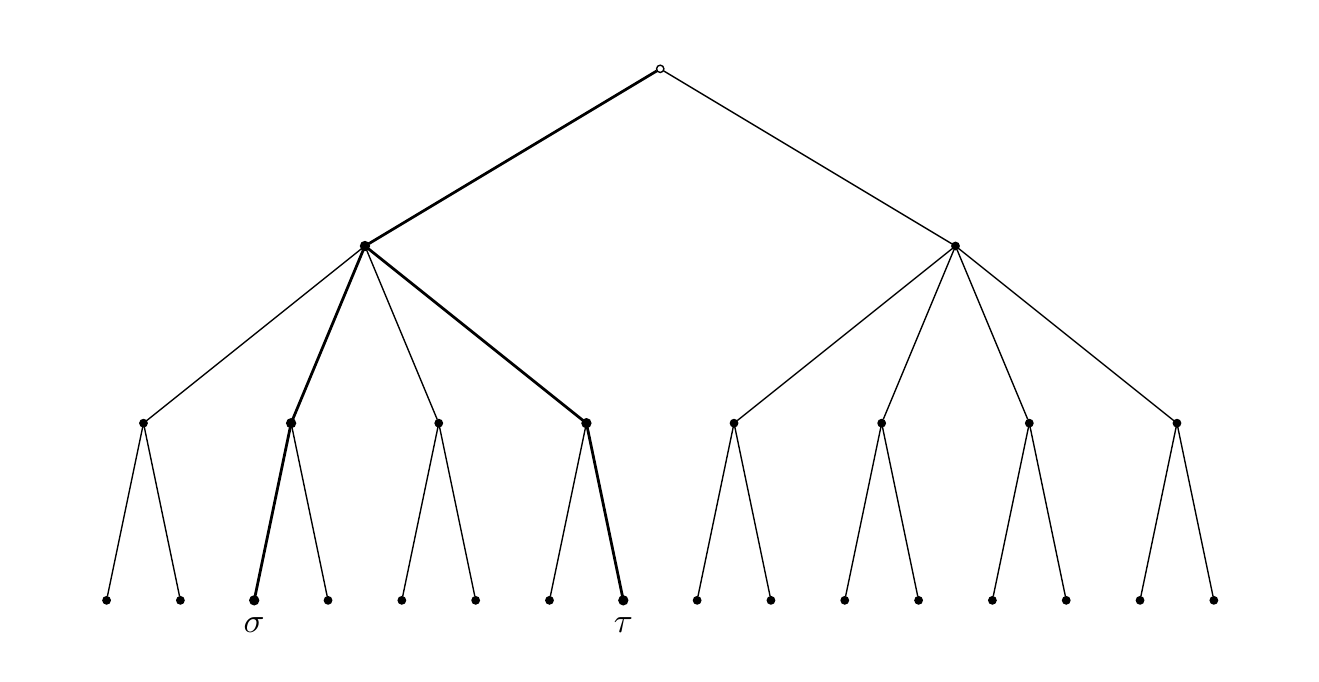}
	\caption{The paths $\sigma$ and $\tau$ are at distance $s_{{\underline{k}}}\left(\sigma,\tau\right)=2\left(a_2+a_3\right)$.}\label{f2}
\end{figure}

Both  the correlations \eqref{corsa} and the metric \eqref{malato} depends
on the vector $\underline k$ and on the assignment of configurations to leaves.
We will discuss soon this.

Like for the Sherrington-Kirkpatrick model, given an inverse temperature $\beta$, we introduce the disorder-dependent partition function
\begin{equation}\label{part-grem}
Z_{\underline k}\left(\beta\right)\coloneqq\sum_{\left\{\sigma\right\}}\ex^{-\beta H_{{\underline{k}}}\left(\sigma\right)}
\end{equation}
and the quenched average of the free energy per site
\begin{equation}\label{free}
F_{\underline k}\left(\beta\right)\coloneqq-\frac{1}{\beta |\underline k|}\,\mathbb{E}\left[\log Z_{\underline k}\left(\beta\right)\right]\,.
\end{equation}
We prove the existence of the thermodynamic limit of \eqref{free} under general assumptions when a parameter $N$ is diverging and the vector $\underline k=\underline{k}\left(N\right)$ is growing in such a way that also $n=n\left(N\right)$ diverges. Contucci et al. \cite{CONT} proved this fact when $n$ is constant. This was obtained applying the same strategy of the Guerra-Toninelli interpolation method \cite{GT}; in particular, they used the inequality in Theorem \ref{loro}. When $n$ is no longer bounded this inequality fails while the inequality in Theorem \ref{nostro} continues to work. We describe now more precisely the growing mechanism of the model and prove the existence of the thermodynamic limit.

\subsubsection{Growing and labeling}

We consider a sequence of growing trees labeled by a sequence of vectors $\underline k\left(N\right)$. For each  $N\in\mathbb{N}$ we have
the tree $\mathcal{T}_{\underline{k}\left(N\right)}$ defined by the following hypothesis and rules.
\newline
\begin{itemize}
\item[(H1)] Let $\left(\alpha_i\right)_{n=1}^{\infty}$ be a sequence of reals larger than 1 satisfying the constraint
\begin{equation}\label{alfa}
\sum_{i=1}^{\infty}\log\alpha_i=\log 2.
\end{equation}
The $\alpha_i$'s define the tree $\mathcal{T}_{\underline{k}\left(N\right)}$ through
\begin{equation}\label{intpart}
k_i\left(N\right)\coloneqq\left\lfloor \frac{N\log \alpha_i}{\log 2}\right\rfloor\,, \qquad i\in\mathbb{N}\,,
\end{equation}
where $\lfloor\, \cdot \,\rfloor$ denotes the integer part.
\\
\item[(H2)] The sequence $\left(a_i\right)_{i=1}^{\infty}$ corresponds to the lengths of the edges from the different levels and the variance of the associated random variables and satisfies the condition $\sum_{i=1}^{\infty}a_i=1$.
\\
\end{itemize}

The exact values of the sums of the series are not really important and could be substituted just by summability conditions. Formula \eqref{intpart} follows by the fact that we ask that the number of edges connecting a given node at level $i-1$ to nodes at level $i$ grows exponentially like $\alpha_i^N$.

Observe that by \eqref{alfa}, for any fixed $N>0$ in ${\underline{k}\left(N\right)}$ just a finite number of components is different from zero. We define \begin{equation}\label{maxx}
n\coloneqq n\left(N\right)\coloneqq \max\{i\,:\, k_i\left(N\right)>0\}
\end{equation} and the finite vector $\underline k\left(N\right)\coloneqq\left(k_1\left(N\right),\ldots,k_n\left(N\right)\right)$. Then, a spin configuration $\sigma\in\left\{-1,1\right\}^{\left|\underline k\left(N\right)\right|}$ is assigned to each leaf. The method is actually arbitrary; indeed, the free energy of the system is obtained summing over all the configurations, thus getting rid of any dependence on the underlying choice.

We assign a spin configuration to each leaf of the tree as follows.
At fixed $N$, we attach to every edge one or more labels of type $\left(m,s\right)$, where $s=\pm 1$  and $m\in\left\{1,\ldots, \left|\underline k(N)\right|\right\}$. Given a leaf there exists a unique path toward the root.
If this path crosses an edge having a label $\left(m,s\right)$ then the configuration $\sigma$ associated to the leave is such that $\sigma\left(m\right)=s$. We assign the label in such a way that every path meets all the labels $m=1, \dots, \left|\underline k\left(N\right)\right|$
and such that different leaves have associated different configurations.

We embed the tree on a plane so that the root is on the top and the paths from the leaves to the root are going upwards. Moreover all the edges connecting a given node with the nodes at the successive level are ordered from left to right. Each edge connecting the level $i-1$ to level $i$ has exactly $k_i\left(N\right)$ labels corresponding to the values $m=\sum_{j=1}^{i-1}k_j\left(N\right)+1, \sum_{j=1}^{i-1}k_j\left(N\right)+2, \dots , \sum_{j=1}^{i}k_j\left(N\right)$. The corresponding values of the parameter $s$ are fixed as follows.

Fix a node at level $i-1$. Number each edge connecting this node with a node at level $i$ with an integer number going from left to right from the value $0$ to $2^{k_i\left(N\right)}-1$. The leftmost will correspond to $0$ while the rightmost to $2^{k_i\left(N\right)}-1$. Do this for each node. Write these integers in binary code so that the leftmost edges are numbered with $k_i\left(N\right)$ zeros and the rightmost with $k_i(N)$ ones. In our setting, the $0$ corresponds to the $-$ sign and the $1$ to the $+$ sign. Then, we associate the lowest value of $m$ to the most significant digit and the highest value of $m$ to the less significant one. See Figure \ref{f3} for an example.

\begin{figure}[htbp!]
	\centering
	\hspace*{-0.4cm}
	\includegraphics{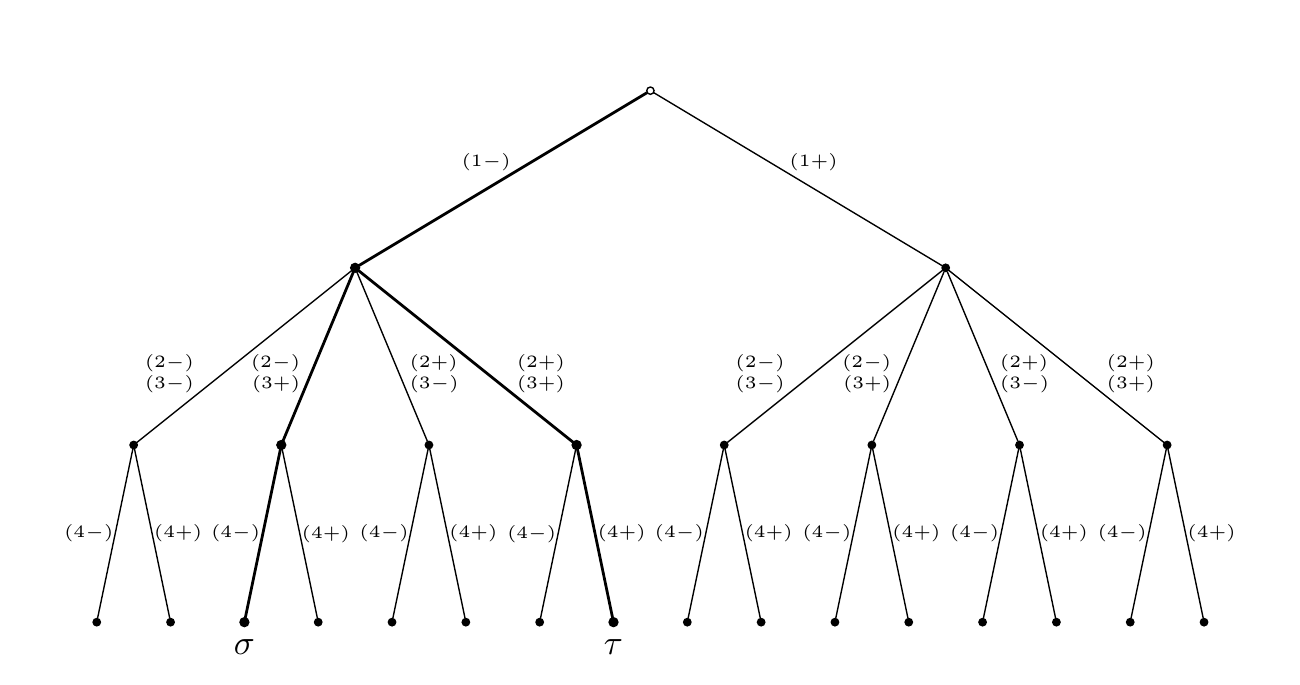}
	\caption{Example of assignation of the labels for $\underline{k}=\left(1,2,1\right)$. Paths $\sigma$ and $\tau$ have spin configurations $\sigma=\left(-1,-1,1,-1\right)$, $\tau=\left(-1,1,1,1\right)$}\label{f3}
\end{figure}

\subsubsection{Splitting the system}

Let $N>0$ and consider a pair of integers $N_1,N_2$ such that $N_1+N_2=N$. We already know how to construct the trees $\mathcal{T}_{\underline{k}\left(N\right)}$, $\mathcal{T}_{\underline{k}\left(N_1\right)}$ and $\mathcal{T}_{\underline{k}\left(N_2\right)}$. Their geometric structure is simply codified by the finite vectors $\underline k\left(N\right)$, $\underline k\left(N_1\right)$, $\underline k\left(N_2\right)$ and we recall that, by definition, we have
\begin{equation}
k_{i}\left(N_j\right)\coloneqq\left\lfloor \frac{N_j\log \alpha_i}{\log 2}\right\rfloor\,, \qquad j=1,2\,, \,\,\,\,\,i\in \mathbb N\,.
\end{equation}
Notice that
\begin{equation}\label{notice}
k_i\left(N_1\right)+k_i\left(N_2\right)\leq k_i\left(N\right)\leq k_i\left(N_1\right)+k_i\left(N_2\right)+1\,.
\end{equation}
We associate the labels to the edges and leaves of the full system $\mathcal{T}_{\underline{k}\left(N\right)}$ as in the previous section. The labels of the two subsystems $\mathcal{T}_{\underline{k}\left(N_1\right)}$ and $\mathcal{T}_{\underline{k}\left(N_2\right)}$ are instead attributed in a slightly different way in order to have different spins (different labels $m$) belonging to the two subsystems.

The labels $m$ attributed to the edges from level $i-1$ to level $i$ in the full system coincide with the set $\left\{\sum_{j=1}^{i-1}k_j\left(N\right)+1,\sum_{j=1}^{i-1}k_j\left(N\right)+2, \dots , \sum_{j=1}^{i}k_j\left(N\right) \right\}$. When we split the system into the two subsystems
we assign to the edges that connect each node in the level $i-1$ to the level $i$ of the subsystem $\mathcal{T}_{\underline{k}\left(N_1\right)}$ the labels
$\left\{\sum_{j=1}^{i-1}k_j\left(N\right)+1, \dots , \sum_{j=1}^{i-1}k_j\left(N\right)+k_i\left(N_1\right) \right\}$
while we assign to the edges that connect each node in the level $i-1$ to the level $i$ of the subsystem $\mathcal{T}_{\underline{k}\left(N_2\right)}$ the labels
$\left\{\sum_{j=1}^{i-1}k_j\left(N\right)+k_i\left(N_1\right)+1,\dots,\sum_{j=1}^{i-1}k_j\left(N\right)+k_i\left(N_1\right)+k_i\left(N_2\right)\right\}$. By \eqref{notice} this is well defined. Once split the labels $m$ into the two subsystems, the assignment of the label $s=\pm$ follow the same rule of the previous section. Since $k_i\left(N_1\right)+k_i\left(N_2\right)$ may be strictly less than $k_i\left(N\right)$, some of the labels $m$ (i.e. some spins) may disappear in the splitting.

\smallskip
We discuss now the behavior of the distances.
Consider two finite vectors $\underline k$ and $\underline k'$ such that
$k'_i\leq k_i$ for any $i$. We assign the labels to $\mathcal T_{\underline k}$ in the usual way while instead we assign the labels to $\mathcal T_{\underline k'}$ as follows.
We assign to the edges that connect each node in the level $i-1$ to the level $i$ of  $\mathcal{T}_{\underline k'}$ arbitrarily $k'_i$ of the $k_i$ labels in $\mathcal T_{\underline k}$. The assignment of the labels $s=\pm$ follows then the usual rule.
	
We call respectively $d_{\underline k}$ and $d_{\underline k'}$ the metrics defined by formula \eqref{malato} for the two trees $\mathcal T_{\underline k}$ and $\mathcal T_{\underline k'}$ and $s_{\underline k}$, $s_{\underline k'}$ the corresponding normalized distances (see \eqref{malato2}). As before given two spin configurations $\sigma,\tau\in \left\{-1,1\right\}^{\left|\underline k\right|}$ we call again $\sigma,\tau\in\left\{-1,1\right\}^{\left|\underline k'\right|}$ the same configurations but restricted just to the labels assigned to the edges in $\mathcal T_{\underline k'}$. We have the following.

\begin{lemma}\label{prp}
Consider two finite vectors $\underline k' \leq \underline k$ and the corresponding trees $\mathcal T_{\underline k}$ and $\mathcal T_{\underline k'}$ with configurations of spins associated to the leaves as above. Then we have
\begin{equation}\label{politica}
s_{\underline k'}\left(\sigma,\tau\right)\leq s_{\underline k}\left(\sigma,\tau\right)\,, \qquad \forall\, \sigma, \tau\,.
\end{equation}
\end{lemma}

\begin{proof}
Consider the tree $\mathcal T_{\underline k}$, two configurations $\sigma, \tau$ associated to two leaves and their corresponding geodetic path. Let us now consider a new finite vector $\underline k'$ obtained by $\underline k$ simply decreasing by one just a single component and preserving all the remaining ones, i.e. $k'_i=k_i-1$
and $k'_j=k_j$ for all $j\neq i$. Suppose that the label $m$ that is missing in $\mathcal T_{\underline k'}$ is $m^*$. The tree $\mathcal T_{\underline k'}$ with the corresponding labeling is obtained from $\mathcal T_{\underline k}$ and the original labeling simply as follows. All the edges connecting nodes at level $i-1$ to nodes at level $i$ in $\mathcal T_{\underline k}$ can be paired into pairs having exactly the same labels apart the one corresponding to $m^*$. The two paired edges will have labels respectively $\left(m^*,+\right)$ and $\left(m^*,-\right)$. If we identify each paired couple of edges, and consequently we identify too the subtrees starting from the identified nodes, we get a tree that coincides with $\mathcal T_{\underline k'}$ with exactly the same assignments of labels. In particular, the leaves associated to $\sigma$, $\tau$ in the new tree will be exactly the original ones after the identification. Finally the geodetic path too remains the same after the identification (see e.g. Figure \ref{f4}).

Since the identification procedure can only shorten this path we have the statement of the lemma when $\underline k'$ is obtained by $\underline k$ decreasing by one just one of its components. We finish the proof observing that any $\underline k'\leq \underline k$ can be obtained by $\underline k$ after a finite numbers of iterations of this type.
\end{proof}

\begin{figure}[htbp!]
	\centering
	\hspace*{-0.175cm}
	\includegraphics{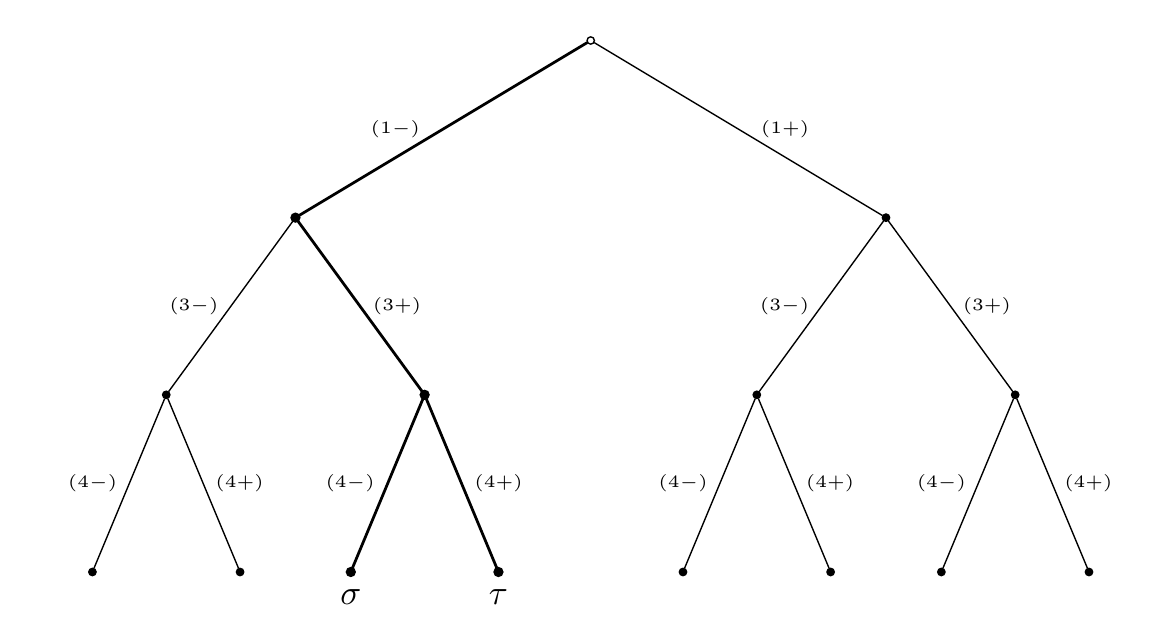}
	\caption{After the coalescence of the branches with $m^*=2$, the configurations $\sigma$ and $\tau$ in Figures \ref{f3} and \ref{f4} are at distance $s_{\underline{k}'}\left(\sigma,\tau\right)=\sqrt{2 a_3}<s_{\underline k}\left(\sigma,\tau\right)=\sqrt{2\left(a_2+a_3\right)}$.}\label{f4}
\end{figure}

\begin{remark}\label{remomarc}
Both $\mathcal T_{\underline k\left(N_i\right)}$, $i=1,2$ are obtained by $\mathcal T_{\underline k\left(N\right)}$ as in the hypothesis of lemma \ref{prp} and we have therefore
\begin{equation}\label{ineq}
s_{\underline k\left(N\right)}\left(\sigma,\tau\right)\ge \max\left\{s_{\underline k\left(N_1\right)}\left(\sigma,\tau\right), s_{\underline k\left(N_2\right)}\left(\sigma,\tau\right)\right\}\,, \qquad \forall \,\sigma, \tau\,,\,\,\,\,\,  i=1,2\,.
\end{equation}
Since by \eqref{notice} we have $\frac{|\underline k(N_1)|+|\underline k(N_2)|}{|\underline k(N)|}\le 1$ we deduce
\begin{equation}
s^2_{\underline k\left(N\right)}\left(\sigma,\tau\right)\ge \frac{\left|\underline k\left(N_1\right)\right|}{\left|\underline k\left(N\right)\right|}s^2_{\underline k\left(N_1\right)}\left(\sigma,\tau\right)+\frac{\left|\underline k\left(N_2\right)\right|}{\left|\underline k\left(N\right)\right|}s^2_{\underline k\left(N_2\right)}\left(\sigma,\tau\right)\,,
\end{equation}
that is equivalent to the super-Pythagorean condition
\begin{equation}
d_{\underline k\left(N\right)}\left(\sigma,\tau\right)\ge \sqrt{d^2_{\underline k\left(N_1\right)}\left(\sigma,\tau\right)+d^2_{\underline k\left(N_2\right)}\left(\sigma,\tau\right)}\,.
\end{equation}

\end{remark}

\subsubsection{Thermodynamic limit}
We define the energy of our sequence of GREM models as
$H_N\left(\sigma\right):= H_{\underline k\left(N\right)}\left(\sigma\right)$ (recall definition \eqref{energia-grem}) and the corresponding partition functions and density of free energies like in  \eqref{part-grem}, \eqref{free} more precisely $Z_N\left(\beta\right)=\sum_{\left\{\sigma\right\}} \ex^{-\beta H_N\left(\sigma\right)}$ and
\begin{equation}\label{dfeg}
F_N(\beta)=-\frac{1}{\beta \left|\underline{k}\left(N\right)\right|}\mathbb E\left[\log Z_N\left(\beta\right)\right]=:\frac{\alpha_N(\beta)}{\beta \left|\underline{k}\left(N\right)\right|}\,,
\end{equation}
where the last equality defines the symbol $\alpha_N(\beta)$.

We need a preliminary Lemma. Let us call $\gamma_i\coloneqq\frac{\log \alpha_i}{\log 2}>0$. Observe that by definition we have $\sum_{i=1}^{+\infty}\gamma_i=1$.

\begin{lemma}\label{senzanome}
	We have
	\begin{equation}
	\lim_{N\to \infty}\frac{\left|\underline k(N)\right|}{N}=1
	\end{equation}
\end{lemma}

\begin{proof}
For any finite $k$ we have
\begin{equation}
 \frac{\sum_{i=1}^{k}\left(N\gamma_i-1\right)}{N}\leq \frac{\left|\underline k\left(N\right)\right|}{N}\leq \frac{\sum_{i=1}^{+\infty}N\gamma_i}{N}\,.
\end{equation}
The right hand side of the above equation is $1$. The left hand side converges when $N\to \infty$ to $\sum_{i=1}^k\gamma_i$. Taking now the limit on $k\to\infty$
we deduce the statement of the Lemma.
\end{proof}

We can now prove the existence of the limit for quenched free energy per site of a GREM model with infinite levels.

\begin{theorem}
Under the hypothesis (H1) and (H2), there exists the limit when $N\to\infty$ of the density of free energy \eqref{dfeg} defined on $\mathcal{T}_{\underline{k}\left(N\right)}$, in the sense that there exists the following limit that coincides with an infimum
	\begin{equation}\label{therm}
	-\infty <\lim_{N\to\infty}-\frac{1}{\beta \left|\underline k\left(N\right)\right|}\,\mathbb{E}\left[\log Z_{N}\left(\beta\right)\right]=\inf_N\frac{\alpha_N(\beta)}{\beta N} <\infty.
	\end{equation}

\end{theorem}

\begin{proof}
We apply the interpolation method for the Gaussian random vectors
$H_{\underline{k}\left(N\right)}\left(\sigma\right)$ and $H_{\underline{k}\left(N_1\right)}\left(\sigma\right)+H_{\underline{k}\left(N_2\right)}\left(\sigma\right)$ that are both labeled by the configurations $\sigma\in \left\{-1,1\right\}^{\left|\underline k\left(N\right)\right|}$. The Gaussian random variables used to compute $H_{\underline{k}\left(N\right)}, H_{\underline{k}\left(N_1\right)}$ and $H_{\underline{k}\left(N_2\right)}$ are all independent among them. Note that since in the splitting some spins are lost then the second Gaussian random vector is degenerate.

We have the following identity
\begin{equation}
\sum_{\left\{\sigma\right\}}\ex^{-\beta H_{\underline{k}\left(N_1\right)}\left(\sigma\right)}\ex^{-\beta H_{\underline{k}\left(N_2\right)}\left(\sigma\right)}=Z_{N_1}\left(\beta\right)Z_{N_2}\left(\beta\right)2^{\left|\underline k\left(N\right)\right|-\left|\underline k\left(N_1\right)\right|-\left|\underline k\left(N_2\right)\right|}\,.
\end{equation}
The last term is due to the fact that some spins may be lost in the splitting.

By Remark \ref{remomarc}, we can apply Theorem \eqref{nostro} getting
\begin{align}\label{extra}
\alpha_N\left(\beta\right)\leq \alpha_{N_1}\left(\beta\right)+\alpha_{N_2}\left(\beta\right)-\Big(\left|\underline k\left(N\right)\right|-\left|\underline k\left(N_1\right)\right|-\left|\underline k\left(N_2\right)\right|\Big)\log 2\,.
\end{align}
Since the last term in the above inequality is nonnegative we obtain that
the sequence $\alpha_N(\beta)$ is subadditive. By Fekete's Lemma we deduce
that there exists the limit
\begin{equation}
\lim_{N\to\infty}\frac{\alpha_N(\beta)}{\beta N}=\inf_N\frac{\alpha_N(\beta)}{\beta N}\,.
\end{equation}
By Lemma \ref{senzanome} we have that
\begin{equation}
\lim_{N\to\infty}\frac{\alpha_N(\beta)}{\beta \left|\underline k\left(N\right)\right|}=
\lim_{N\to\infty}\frac{\alpha_N(\beta)}{\beta N}\,,
\end{equation}
and we get the main statement of the Theorem.

\smallskip

It remains just to prove that the limit
is strictly bigger than $-\infty$.

This follows by the summability of the variances $a_i$'s. Indeed, we prove that for any $N>0$, $-\beta F_{N}\left(\beta\right)$ is bounded from above. We have
\begin{eqnarray}
-\beta F_{N}\left(\beta\right)&=&\frac{1}{\left|\underline k\left(N\right)\right|}\,\mathbb{E}\left[\log Z_{N}\left(\beta\right)\right]=\frac{1}{\left|\underline k\left(N\right)\right|}\,\mathbb{E}\left[\log\sum_{\left\{\sigma\right\}}\ex^{\beta(\ve^{(\sigma)}_1+\ve^{(\sigma)}_2+\ldots+\ve^{(\sigma)}_{n\left(N\right)})}\right]
\nonumber\\
&\leq &\frac{1}{\left|\underline k\left(N\right)\right|}\,\log \sum_{\left\{\sigma\right\}} \mathbb{E}\left[\ex^{\beta(\ve^{(\sigma)}_1+\ve^{(\sigma)}_2+\ldots+\ve^{(\sigma)}_{n\left(N\right)})}
\right]\,,
\end{eqnarray}
where we used Jensen's inequality.
Since the $\ve_i^{(\sigma)}$ are independent, the expectation value in the last row is the product of generating functions:
\begin{equation}
\mathbb{E}\left[\ex^{\beta\ve_i^{(\sigma)}}\right]=\ex^{\frac{\beta}{2}a_i}\qquad \forall \,i\,,
\end{equation}
hence
\begin{align}
-\beta F_{N}\left(\beta\right)\leq \frac{1}{\left|\underline k\left(N\right)\right|}\,\log \sum_{\left\{\sigma\right\}}
\ex^{\frac{\beta}{2}\sum_{i=1}^{n\left(N\right)}a_i}
\le \frac{1}{\left|\underline k\left(N\right)\right|}\,\log \sum_{\left\{\sigma\right\}}\ex^{\frac{\beta}{2}}=\log 2 + \frac{\beta}{2\left|\underline k\left(N\right)\right|}\,,
\end{align}
where we used the fact that $\sum_{i=1}^{\infty}a_i=1$.
\end{proof}

Just as a remark we show in Lemma \ref{fine} in the Appendix that the third term
in the right hand side of \eqref{extra} is negligible when $N$ is large.
This fact is irrelevant for the proof but it is interesting in itself since for different models we could have a similar situation but with the wrong sign and a bound of this type could allow to apply the generalized subadditive lemmas in \cite{dbe}.

\section*{Acknowledgments}
We thank Adriano Barra, Francesco Guerra and Fabio Lucio Toninelli for several useful comments, suggestions and remarks.

\section*{Appendix}

In this Appendix we collect for the reader's convenience the proofs of some elementary Lemmas.

\begin{lemma}
	Let $v^{(1)},\dots, v^{(n)}$ and $w^{(1)},\dots, w^{(n)}$ be two
	collections of n vectors in $\mathbb R^n$. We have that
	\begin{equation}
	\big( v^{(i)},v^{(j)}\big) =\big( w^{(i)},w^{(j)}\big)\,,
	\qquad \forall \,i,j\, \label{cosug}
	\end{equation}
	if and only if there exists $O\in O\left(n\right)$ such that $w^{(i)}=Ov^{(i)}$ for any
	$i$. \label{coseni}
\end{lemma}
\begin{proof}
	
	If there exists $O\in O\left(n\right)$ such that $w^{(i)}=Ov^{(i)}$ for any
	$i$ the \eqref{cosug} is clearly true. Let us prove the converse statement.
	
	Suppose firstly that the vectors $v^{(i)}$ are linearly independent. Then given any $x\in
	\mathbb R^n$ we can write $x=\sum_ic^x_iv^{(i)}$ for some
	coefficients $c^x_i$. In this case we define the matrix $O$ by
	$$
	Ox\coloneqq \sum_{i=1}^nc^x_iOv^{(i)}=\sum_{i=1}^nc^x_iw^{(i)}\,.
	$$
	We need only to check that it is indeed orthogonal. This follows by
	\begin{eqnarray}
	\left(Ox,Oy\right)= \sum_{i,j}c^x_ic^y_j\big( w^{(i)},w^{(j)}
	\big)=\sum_{i,j}c^x_ic^y_j\big( v^{(i)},v^{(j)}\big)=\left(
	x,y\right)\,.
	\end{eqnarray}
	If the vectors $v^{(i)}$ are not linearly independent, without loss
	of generality we assume that $v^{(1)},\dots , v^{(k)}$ is a maximal
	collection of linearly independent vectors. Using \eqref{cosug}
	we deduce that  $w^{(1)},\dots , w^{(k)}$ too is a maximal
	collection of linearly independent vectors among $w^{(1)},\dots ,
	w^{(n)}$. Indeed if for example we have $v^{(l)}=\alpha
	v^{(i)}+\beta v^{(j)}$ then we have
	\begin{equation}
	\big| w^{(l)}-\alpha w^{(i)}-\beta w^{(j)}\big|^2 =
	 \big| v^{(l)}-\alpha v^{(i)}-\beta v^{(j)}\big|^2 =0\,,\label{bus}
	\end{equation}
	that implies
	$w^{(l)}=\alpha w^{(i)}+\beta w^{(j)}$.

	Let us call $\tilde{v}^{(1)},\dots,\tilde{v}^{(n)}$ a basis whose
	vectors are defined as follows. We have that
	$\tilde{v}^{(i)}=v^{(i)}$ for $i=1,\dots,k$. The remaining vectors
	$\tilde{v}^{(k+1)},\dots,\tilde{v}^{(n)}$ are an orthonormal basis
	of the orthogonal complement of $\mathrm{Span}\left\{v^{(1)},\dots,
	v^{(k)}\right\}$. Likewise we call $\tilde{w}^{(1)},\dots,
	\tilde{w}^{(n)}$ a basis whose vectors are defined as follows. We
	have that $\tilde{w}^{(i)}=w^{(i)}$ for $i=1,\dots,k$. The remaining
	vectors $\tilde{w}^{(k+1)},\dots,\tilde{w}^{(n)}$ are an orthonormal
	basis of the orthogonal complement of $\mathrm{Span}\left\{w^{(1)},\dots,
	w^{(k)}\right\}$.
	
	Given any $x\in \mathbb R^n$ we have $x=\sum_ic^x_i\tilde{v}^{(i)}$ and
	we define the matrix $O$ by setting
	$$
	Ox\coloneqq \sum_ic^x_iO\tilde{v}^{(i)}=\sum_ic^x_i\tilde{w}^{(i)}\,.
	$$
	Since
	$$
	\big( \tilde{w}^{(i)},\tilde{w}^{(j)}\big) =\big(
	\tilde{v}^{(i)},\tilde{v}^{(j)}\big) \,, \qquad \forall \,i,j
	$$
	we obtain as before that $O$ is orthogonal. With an argument similar
	to the one in \eqref{bus} it is easy to obtain that
	$w^{(i)}=Ov^{(i)}$.
\end{proof}

\begin{lemma}
	Let $v^{(1)},\dots,v^{(n)}$ and $w^{(1)},\dots,w^{(n)}$ be two
	collections of vectors in $\mathbb R^n$. We have that
	\begin{equation}
	\big|v^{(i)}-v^{(j)}\big|=\big|w^{(i)}-w^{(j)}\big|\,, \qquad  \forall \,i,j\,,
	\label{pimpa}
	\end{equation}
	If and only if there exists $O\in O\left(n\right)$ and a vector $b\in \mathbb R^n$ such
	that
	\begin{equation}
	w^{(i)}=Ov^{(i)}+b\,, \qquad \forall i\,. \label{bussole}
	\end{equation}
	\label{pimpalemma}
\end{lemma}
\begin{proof}
	If \eqref{bussole} holds then we have clearly \eqref{pimpa}. It remains to prove the converse statement.

	First we prove the Theorem under the additional assumption
	$\sum_iv^{(i)}=\sum_iw^{(i)}=0$. We have
	\begin{eqnarray}
	& & \big| v^{(i)}\big|^2+\big| v^{(j)} \big|^2
	-2\,\big( v^{(i)},v^{(j)}\big) \nonumber \\
	& &=
	\big|v^{(i)}-v^{(j)}\big|^2=\big| w^{(i)}-w^{(j)}\big|^2\nonumber \\
	& & =\big| w^{(i)}\big|^2+\big| w^{(j)} \big|^2
	-2\,\big( w^{(i)},w^{(j)}\big)\,.\label{nonnabruna}
	\end{eqnarray}
	Summing over $j$ the two extremal sides in \eqref{nonnabruna} we get
	\begin{equation}
	n\big| v^{(i)}\big|^2+\sum_{l=1}^n\big| v^{(l)}\big|^2=
	n\big| w^{(i)}\big|^2+\sum_{l=1}^n\big| w^{(l)}\big|^2\,, \qquad \forall\, i\,.
	\label{sfleggo}
	\end{equation}
	From \eqref{sfleggo} we deduce
	\begin{equation}
	\big| v^{(i)}\big|^2= \big| w^{(i)}\big|^2\,, \qquad \forall\, i\,.
	\label{svenuto}
	\end{equation}
	Indeed  summing over $i$ both sides of
	\eqref{sfleggo} we get
	$$
	\sum_{l=1}^n\big| v^{(l)}\big|^2= \sum_{l=1}^n\big|
	w^{(l)}\big|^2\,,
	$$
	that inserted  in \eqref{sfleggo} gives \eqref{svenuto}. Using
	\eqref{svenuto} in \eqref{nonnabruna} we deduce
	\begin{equation}
	\big( v^{(i)},v^{(j)}\big)=
	\big( w^{(i)},w^{(j)} \big)\,, \qquad \forall \,i,j\,. \label{vicino}
	\end{equation}
	The validity of \eqref{bussole} with $b=0$ now follows from Lemma
	\ref{coseni}.

	The general case is then obtained as follows. Define the vectors
	$\tilde{v}^{(i)}\coloneqq v^{(i)}-\frac 1n \sum_j v^{(j)}$ and
	$\tilde{w}^{(i)}\coloneqq w^{(i)}-\frac 1n \sum_j w^{(j)}$. Then we have
	$\sum_i\tilde{v}^{(i)}=\sum_i\tilde{w}^{(i)}=0$ and moreover
	$\left|\tilde{v}^{(i)}-\tilde{v}^{(j)}\right|=\left|\tilde{w}^{(i)}-\tilde{w}^{(j)}\right|$.
	From the result just proved we know that there exists $O\in O\left(n\right)$
	such that $\tilde{w}^{(i)}=O\tilde{v}^{(i)}$. Using linearity and
	the definitions of the vector we obtain \eqref{bussole} with
	$$
	b=\frac 1n \sum_iw^{(i)}-\frac 1n \sum_iOv^{(i)}\,.
	$$
\end{proof}

We show here that the extra terms in \eqref{extra} are indeed negligible when $N$ is large.

\begin{lemma}\label{fine}
	We have
	\begin{equation}
	\lim_{N\to \infty}\sup_{N_1+N_2=N}\frac{\left|\underline k\left(N\right)\right|-\left|\underline k\left(N_1\right)\right|-\left|\underline k\left(N_2\right)\right|}{|\underline k(N)|}=0\,.
	\end{equation}
\end{lemma}

\begin{proof}
	Let us call $I\left(N\right)\subseteq \mathbb N$ the set $I\left(N\right)\coloneqq\left\{i\,:\, k_i\left(N\right)>0\right\}$. We define also
	$J\left(N\right)\coloneqq I\left(N\right)\cap \left(I\left(N-1\right)\right)^C$.
	We have for any $N_1+N_2=N$ that
	\begin{equation}\label{critto}
	\left|\underline k\left(N\right)\right|-\left|\underline k\left(N_1\right)\right|-\left|\underline k\left(N_2\right)\right|\leq |I(N)|
	\end{equation}
	By definition we have that if $i\in J\left(N\right)$ then $\frac 1N\leq \gamma_i<\frac{1}{N-1}$. We have therefore
	\begin{equation}\label{equi}
	1=\sum_{i=1}^{+\infty}\gamma_i\geq \sum_{\ell=1}^{+\infty}\left|J\left(\ell\right)\right|\frac 1\ell\,.
	\end{equation}
	We deduce therefore that the series on the right hand side has to be convergent. Since $|I(N)|=\sum_{\ell=1}^N|J(\ell)|$
	we have
	\begin{equation}
	\frac{|I(N)|}{N}=\sum_{\ell=1}^N\frac{\ell}{N}|J(\ell)|\frac 1\ell\,.
	\end{equation}
	The series on the right hand side of \eqref{equi} is convergent, thus we deduce the statement by the dominated convergence Theorem. This, together with \eqref{critto} and Lemma \ref{senzanome}, concludes the proof.
\end{proof}

\end{document}